\documentclass{article}

\usepackage{ijcai20}

\usepackage{times}
\usepackage{soul}
\usepackage{url}
\usepackage[hidelinks]{hyperref}
\usepackage[utf8]{inputenc}
\usepackage[small]{caption}
\usepackage{graphicx}
\usepackage{amsmath}
\usepackage{amsthm}
\usepackage{booktabs}
\usepackage{algorithm}
\usepackage{algorithmic}
\usepackage{todonotes}

\usepackage{bm}
\usepackage{xfrac}
\usepackage{amssymb}
\usepackage{amsfonts}
\usepackage{enumitem}
\usepackage{makecell}
\usepackage{multicol}
\usepackage{multirow}
\usepackage{subcaption}

\newtheorem{lemma}{Lemma}
\newtheorem{theorem}{Theorem}
\newtheorem{definition}{Definition}

\usepackage{xspace}

\newcommand{\FIXED}{\textsc{Fixed}\xspace}
\newcommand{\RANDOM}{\textsc{Random}\xspace}
\newcommand{\STRATEGIC}{\textsc{Strategic}\xspace}

\usepackage{tcolorbox}

\newcommand{\citet}[1]{\citeauthor{#1}~\shortcite{#1}}

\title{Robust Market Making via Adversarial Reinforcement Learning}
\author{
Thomas Spooner\And
Rahul Savani
\affiliations
Department of Computer Science, University of Liverpool\\
\emails
\{t.spooner, rahul.savani\}@liverpool.ac.uk
}

\begin{document}

\maketitle

\begin{abstract}

    We show that adversarial reinforcement learning (ARL) can be used to
    produce market marking agents that are \emph{robust} to \emph{adversarial}
    and adaptively-chosen market conditions. To apply ARL, we turn the
    well-studied single-agent model of \citeauthor{avellaneda2008high}
    \shortcite{avellaneda2008high} into a discrete-time zero-sum game between a
    market maker and adversary.  The adversary acts as a proxy for other market
    participants that would like to profit at the market maker's expense.  We
    empirically compare two conventional single-agent RL agents with ARL, and
    show that our ARL approach leads to: 1) the emergence of risk-averse
    behaviour without constraints or domain-specific penalties; 2) significant
    improvements in performance across a set of standard metrics, evaluated
    with \emph{or without} an adversary in the test environment, and; 3)
    improved robustness to model uncertainty. We empirically demonstrate that
    our ARL method consistently converges, and we prove for several special
    cases that the profiles that we converge to correspond to Nash equilibria
    in a simplified single-stage game.

\end{abstract}

\section{Introduction}

Market making refers to the act of providing liquidity to a market by
continuously quoting prices to buy and sell a given
financial instrument. The difference between these prices is called a
\emph{spread}, and the goal of a market maker (MM) is to repeatedly earn a spread
by transacting in both directions.
However, this is not done without risk. Market makers expose themselves to
\emph{adverse selection}, where toxic agents exploit their technological or
informational edge, transacting with and thereby changing 
the MM's inventory before an adverse price move causes them a loss.
This is known as inventory risk, and has
been the subject of a great deal of research in optimal control, artificial
intelligence, and reinforcement learning (RL) literature.

A standard assumption in existing work has been that the MM
has perfect knowledge of market conditions. Robustness of market making
strategies to model ambiguity has only recently received attention - 
\citet{cartea2017algorithmic} extend optimal control approaches for the market
making problem to address the risk of model misspecification. In this paper, 
we deal with this type of risk by taking an
\emph{adversarial} approach.
We design market making agents that are robust to adversarial and
adaptively chosen market conditions by applying adversarial RL.  Our starting
point is a well-known single-agent mathematical model of market
making of \citet{avellaneda2008high}, which has been used extensively in the
quantitative finance
~\cite{cartea2015algorithmic,cartea2017algorithmic,gueant2013dealing,gueant2017optimal}.
We convert this into a discrete-time game, with a new ``market player'', the
adversary, that can be thought of as a proxy for other market participants that
would like to profit at the expense of the market maker. The adversary controls
dynamics of the market environment in a zero-sum game against the market maker.

We evaluate the sensitivity of RL-based strategies to three core parameters of
the market model dynamics that affect prices and execution, where each of these parameters
naturally varies over time in real markets. We thus go beyond
the fixed parametrisation of existing models --- henceforth called the
\textbf{\FIXED} setting --- with two extended learning settings. The
\textbf{\RANDOM} setting initialises each instance of the model (that is, episode 
in RL terminology) with values for the three parameters that are 
chosen independently and uniformly at random from three appropriate intervals. The
\textbf{\STRATEGIC} setting features an adversary -- an independent ``market'' learner whose
objective is to \emph{choose the parameters} from these same intervals so as to
\emph{minimise} the cumulative reward of the market maker in a zero-sum game. The
\RANDOM and \STRATEGIC settings are, on the one hand, more realistic than
the \FIXED setting, but on the other hand, significantly more complex for the
market making agent. We show that market making strategies trained in each of
these settings yield significantly different behaviour, and we demonstrate
striking benefits of our \STRATEGIC setting.

\subsection{Contributions}
The key contributions of this paper are as follows:
\begin{itemize}[leftmargin=0.35cm]
    \item We introduce a game-theoretic adaptation of a standard
        mathematical model of market making. Our adaption is 
        useful to train robust MMs, and evaluate their performance in the
        presence of epistemic risk (Sections~\ref{sec:model}~and~\ref{sec:one_shot}).

    \item We propose an algorithm for adversarial reinforcement learning in the
        spirit of RARL~\cite{pinto2017robust}, and demonstrate its
        effectiveness in spite of the well known challenges
        associated with finding Nash equilibria of stochastic
        games~\cite{littman1994markov}
        (Sections~\ref{sec:stochastic_games}~and~\ref{sec:experiments}).

    \item We thoroughly investigate the impact of three environmental settings
        (one adversarial) on learning market making.  We show that training
		against a \STRATEGIC adversary \emph{strictly dominates} the other
		two settings (\FIXED and \RANDOM)
        in terms of a set of standard desiderata, including the
        Sharpe ratio (Section~\ref{sec:experiments}).

    \item We prove that, in several key instances of the \STRATEGIC setting,
        the single-stage instantiation of our game has a Nash
        equilibrium resembling that found by our ARL algorithm for the
		multi-stage game. We then confirm broader existence of (approximate)
		equilibria in the multi-stage game by empirical best response
		computations (Sections~\ref{sec:one_shot} and~\ref{sec:experiments}).

\end{itemize}

\subsection{Related Work}\label{sec:related_work}

\paragraph{Optimal control and market making.}
The theoretical study of market making originated from the pioneering work of
\citet{ho1981optimal},
\citet{glosten1985bid} and
\citet{grossman1988liquidity}, among
others. Subsequent work focused on characterising optimal behaviour under
different market dynamics and contexts. Most relevant is the work of 
\citet{avellaneda2008high}, who
incorporated new insights into the dynamics of the limit order book to give a
new market model, which is the one used in this paper. They derived closed-form
expressions for the optimal strategy of an MM with an exponential utility function when
the MM has perfect knowledge of the model and its parameters. 
This same problem was then studied for other utility functions~%
\cite{fodra2012high,gueant2013dealing,cartea2015algorithmic,cartea2015order,gueant2017optimal}.
As mentioned above, \citet{cartea2017algorithmic}
study the impact of uncertainty in the model of
\citeauthor{avellaneda2008high}~\shortcite{avellaneda2008high}: they drop the
assumption of perfect knowledge of market dynamics, and consider how an MM
should optimally trade while being robust to possible misspecification.  This type of
\emph{epistemic risk} is the primary focus of our paper.

\paragraph{Machine learning and market making.}
Several papers have applied AI techniques to design automated market makers for
financial markets.\footnote{A separate strand of work in AI and Economics and
    Computation has studied automated market makers for prediction markets, see
    e.g.,~\citeauthor{othman}~\shortcite{othman}. While some similarities to
the financial market making problem pertain, the focus in that strand of work
focusses much more on price discovery and information aggregation.}
\citeauthor{chan2001electronic}~\shortcite{chan2001electronic} focussed on the
impact of noise from uninformed traders on the quoting behaviour of a market
maker trained with reinforcement learning.
\citeauthor{AbernethyK13}~\shortcite{AbernethyK13} used an \emph{online
learning} approach.
\citeauthor{spooner2018market}~\shortcite{spooner2018market} later explored the
use of reinforcement learning to train inventory-sensitive market making agents
in a fully data-driven limit order book model. Most recently,
\citeauthor{gueant2019deep}~\shortcite{gueant2019deep} addressed scaling issues
of finite difference approaches for high-dimensional, multi-asset market making
using model-based RL.
While the approach taken in this paper is also based on RL, unlike the majority
of these works, our underlying market model is taken from the mathematical
finance literature. There, models are typically analysed using methods from
optimal control. We justify this choice in Section~\ref{sec:model}.  To the
best of our knowledge, we are the first to apply ARL to derive trading
strategies that are robust to epistemic risk.

\paragraph{Risk-sensitive reinforcement learning.}
Risk-sensitivity and safety in RL has been a highly active topic
for some time. This is especially true in robotics where exploration is very
costly. For example, \citeauthor{tamar2012policy}~\shortcite{tamar2012policy}
studied policy search in the presence of variance-based risk criteria, and
\citeauthor{bellemare2017distributional}~\shortcite{bellemare2017distributional}
presented a technique for learning the full distribution of (discounted)
returns; see
also~\citeauthor{garcia2015comprehensive}~\shortcite{garcia2015comprehensive}.
These techniques are powerful, but can be complex to implement and can suffer
from numerical instability. This is especially true when using exponential
utility functions which, without careful consideration, may diverge early in
training due to large negative rewards~\cite{madisson2017particle}.
An alternative approach is to train agents in an adversarial
setting~\cite{pinto2017robust,perolat2018actor} in the form of a zero-sum game.
These methods tackle the problem of epistemic risk\footnote{The problem of
    robustness has also been studied outside of the use of adversarial
learning; see, e.g.,~\cite{rajeswaran2016epopt}.} by explicitly accounting for
the misspecification between train- and test-time simulations. This robustness
to test conditions and adversarial disturbances is especially relevant in
financial problems and motivated the approach taken in this paper.

\section{Trading Model}\label{sec:model}

We consider a standard model of market making as introduced by
\citeauthor{avellaneda2008high}~\shortcite{avellaneda2008high} and
studied by many others including
\citeauthor{cartea2017algorithmic}~\shortcite{cartea2017algorithmic}. The MM
trades a single asset for which the price, $Z_n$, evolves stochastically. 
In discrete-time,
\begin{equation}\label{eq:midprice}
    Z_{n+1} = Z_n + b_n \Delta t + \sigma_n W_n,
\end{equation}
where $b_n$ and $\sigma_n$ are the \emph{drift} and \emph{volatility}
coefficients, respectively. The randomness comes from the sequence of
independent, Normally-distributed random variables, $W_n$, each with mean zero
and variance $\Delta t$. The process begins with initial value $Z_0 = z$ and
continues until step $N$ is reached.

The market maker interacts with the environment at each step by placing limit
orders around $Z_n$ to buy or sell a single unit of the asset. The prices at which the MM
is willing to buy (bid) and sell (ask) are denoted by $p_n^+$ and $p_n^-$,
respectively, and may be expressed as offsets from $Z_n$:
\begin{equation}\label{eq:offsets}
    \delta_n^\pm = \pm [Z_n - p_n^\pm] \geq 0;
\end{equation}
these may be updated at each timestep at no cost to the agent. Equivalently, we
may define:
\begin{equation}
    \begin{aligned}
        \psi_n &= \delta_n^+ + \delta_n^- > 0, \\
        p_n &= \frac{1}{2}(p_n^+ + p_n^-) =
            Z_n + \frac{1}{2}(\delta_n^- - \delta_n^+),
    \end{aligned}
\end{equation}
called the \emph{quoted spread} and \emph{reservation price}, respectively.
These relate to the agent's need for immediacy and bias in execution, amongst
other things.

In a given time step, the probability that one or both of the agent's limit
orders are executed depends on the liquidity in the market and the values
$\delta_n^\pm$. 
Transactions occur when market orders, which arrive at random times, have
sufficient size to consume one of the agent's limit orders. 
These interactions are modelled by independent
Poisson processes, denoted by $N_n^+$ and $N_n^-$ for the bid and ask sides,
respectively, with intensities $\lambda_n^\pm$; not to be confused with the
terminal timestep $N$. The dynamics of the agent's inventory process, or
\emph{holdings},~$H_n$, are then captured by the difference between these two
terms,
\begin{equation}
    H_n = (N_n^+ - N_n^-) \in [\underline{H}, \overline{H}],
\end{equation}
where $H_0$ is known and the values of $H_n$ are constrained so that trading
stops on the opposing side of the book when either limit is reached. 
The order arrival intensities are given by:
\begin{equation}\label{eq:order_arrival}
    \lambda_n^\pm = A_n^\pm e^{-k_n^\pm \delta_n^\pm},
\end{equation}
where $A_n^\pm, k_n^\pm > 0$ describe the \emph{rate of arrival of market orders} and
\emph{distribution of volume in the book}, respectively. This particular form
derives from assumptions and observations on the structure and behaviour of
limit order books which we omit here for simplicity;
see~\citet{avellaneda2008high}, \citet{gould2013limit}, and \citet{abergel2016limit} for more details.

In this framework, the evolution of the market maker's cash is given by the
difference relation,
\begin{equation}
    \begin{split}
        X_{n+1} &= X_n + p_n^- \Delta N_n^- - p_n^+ \Delta N_n^+, \\
            &= X_n + \delta_n^- \Delta N_n^- + \delta_n^+ \Delta N_n^+ - Z_n
            \Delta H_n,
    \end{split}
\end{equation}
where $\Delta N_n^\pm \equiv N_{n+1}^\pm - N_n^\pm$. The cash flow
is a combination of: the profit due to executing at the premium $\delta_n^\pm$,
and the change in value of the agent's holdings. The total value accumulated by
the agent by timestep $n$ may thus be expressed as the sum of the cash held and
value invested: $\Pi(X_n, H_n, Z_n)$, where
\begin{equation}\label{eq:mtm}
    \Pi(X, H, Z) = X + HZ,
\end{equation}
and $\Pi_n \equiv \Pi(X_n, H_n, Z_n)$. This is known as the
\emph{mark-to-market} (MtM) value of the agent's portfolio.

\paragraph{Why not use a data driven approach?}
Previous research into the use of RL for market making ---
and trading more generally --- has focussed on data-driven limit order book
models; see \citet{nevmyvaka2006reinforcement}, \citet{spooner2018market}, and
\citet{vyetrenko2019risk}. 
These methods, however, are not amenable to the type of analysis presented in
Section~\ref{sec:one_shot}. Using an analytical model allows us to examine the
characteristics of adversarial training in isolation while minimising systematic
error due to bias often present in historical data.

\section{Game Formulation and Single-Stage Analysis}
\label{sec:one_shot}

We use the market dynamics above to define a zero-sum
stochastic game between a market maker and an adversary
that acts as a proxy for all other market participants.

\begin{definition}\label{def:adv_mm}[Market Making Game]
    The game between MM and an adversary has $N$ stages.  At each stage,
    MM chooses~$\delta^\pm$ and the adversary $\{b, A^\pm, k^\pm\}$. The
    resulting stage payoff is given by expected change in MtM value of the
    MM's portfolio, i.e., $\mathbb{E}[\Delta\Pi]$, see \eqref{eq:mtm}. The
    total payoff paid by the adversary to MM is the sum of the stage payoffs.
\end{definition}

In the remainder of this section, we study theoretically the game when $N=1$,
i.e., when there is a single stage. Later, in Section~\ref{sec:experiments}, we
analyse empirically the game for $N=200$.

\paragraph{Single-stage analysis.}
At each stage, the MM's payoff may be unrolled to give:
\begin{equation}\label{eq:payoff}
    f(\delta^\pm; b, A^\pm, k^\pm) = \lambda^+\cdot(\delta^+ + b) +
    \lambda^-\cdot(\delta^- - b) + bH.
\end{equation}
For certain parameter ranges, this equation is concave in $\delta^\pm$.

\begin{lemma}[Payoff Concavity in $\delta^\pm$]\label{lem:concavity}
    The payoff~\eqref{eq:payoff} is a concave function of $\delta^\pm$ on the
    intervals $[0, \sfrac{2}{k} - b]$, and $[0, \sfrac{2}{k} + b]$,
    respectively.
\end{lemma}
\begin{proof}
    The first derivative of the payoff w.r.t.\ $\delta^\pm$ is given by:
    \begin{equation}\label{eq:f_d1}
        \frac{\partial f}{\partial \delta^\pm} = \lambda^\pm \left[1 -
        k(\delta^\pm \mp b) \right].
    \end{equation}
    The Hessian matrix is thus given by,
    \begin{equation*}
        \left\{\begin{matrix} k^+ \lambda^+
        \left[k^+(\delta^+ + b) - 2\right] & 0 \\ 0 & k^- \lambda^-
        \left[k^-(\delta^- - b) - 2\right] \end{matrix}\right\},
    \end{equation*}
    which is negative semi-definite iff $\delta^\pm \leq \frac{2}{k^\pm} \mp
    b$.
\end{proof}

Note that~\eqref{eq:payoff} is linear in both $b$ and $A^\pm$. From this, we
next show that there exists a Nash equilibrium (NE) when $\delta^\pm$ and~$b$
are controlled, and $k^\pm$ and $A^\pm$ are fixed.

\begin{theorem}[NE for fixed $A^\pm, k^\pm$]\label{thm:one_shot_nash}
    There is a pure strategy Nash equilibrium $(\delta^\pm_\star,
    b_\star)$ for $(\delta^+, \delta^-) \in [0, \sfrac{2}{k} - b] \times [0,
    \sfrac{2}{k} + b]$ and $b \in [\underline{b}, \overline{b}]$ (with finite
    $\underline{b},\overline{b}$),
    \begin{equation}\label{eq:one_shot_deltas}
        \delta_\star^\pm = \frac{1}{k} \pm b_\star; \quad
        b_\star = \begin{cases}
            \underline{b} &\quad H > 0, \\
            \overline{b} &\quad H < 0,
        \end{cases}
    \end{equation}
    which is unique for $|H| > 0$. When $H = 0$, there is an equilibrium for
    every value $b_\star \in [\underline{b}, \overline{b}]$.
\end{theorem}
\begin{proof}
    Equating \eqref{eq:f_d1} to zero and solving gives
    \eqref{eq:one_shot_deltas}. To prove that these correspond to a pure
    strategy Nash Equilibrium of the game we show that the payoff is
    quasi-concave (resp.\ quasi-convex) in the MM's (resp.\ adversary's)
    strategy and then apply Sion's minimax theorem~\cite{sion1958general}.
    This follows from Lemma~\ref{lem:concavity} for MM, and by the linearity of
    the adversary's payoff w.r.t.\ $b$. For $|H| > 0$, there is a unique
    solution.  When $H = 0$, there exists a continuum of solutions, all with
    equal payoff.
\end{proof}

The solution~\eqref{eq:one_shot_deltas} has a similar form to that of the
optimal strategy for linear utility with terminal inventory
penalty~\cite{fodra2012high}, or equivalently that of a myopic agent with
running penalty~\cite{cartea2015order}. Interestingly, the extension of
Theorem~\ref{thm:one_shot_nash} to an adversary with control over all three
model parameters $\{b, A^\pm, k^\pm\}$ yields a similar result. In this case we
refer the reader the extended version of the paper~\cite{spooner2020robust} for
a proof, and leave uniqueness to future work.


\begin{theorem}[NE for general case]
    There exists a pure strategy Nash equilibrium $(\delta^\pm_\star,
    \{b_\star, A^\pm_\star, k^\pm_\star\})$ of the MM game
    for~\eqref{eq:one_shot_deltas}, $A^\pm_\star = \underline{A}$ and
    $k^\pm_\star = \overline{k}$.
\end{theorem}

\section{Adversarial Training}
\label{sec:stochastic_games}

The single-stage setting is informative but unrealistic.  Next we investigate a
range of multi-stage settings with various different restrictions on the
adversary, and explore how adversarial training can improve the robustness of MM
strategies.  The following three types of adversary in turn increase the
freedom of the adversary to control the market's dynamics:

\begin{description}[wide,itemindent=\labelsep]
    \item[\FIXED.] The simplest possible adversary always plays the
        same fixed strategy, for us: $b_n = 0$, $A_n^\pm = 140$ and $k_n^\pm = 1.5$;
        these values match those originally used by
        \citeauthor{avellaneda2008high}~\shortcite{avellaneda2008high}. This
        amounts to a \emph{single-agent learning setting with stationary
        transition dynamics}.
    \item[\RANDOM.] The second type of adversary instantiates each episode
        with parameters chosen independently and uniformly at random from the
        following ranges: $b_n = b \in [-5, 5]$, $A_n^\pm = A \in [105, 175]$
        and $k_n^\pm = k \in [1.125, 1.875]$. These are chosen at the start of
        each episode and remain fixed until the terminal timestep $N$. This is
        analogous to \emph{single-agent RL with non-stationary transition
        dynamics}.
    \item[\STRATEGIC.] The final type of adversary chooses the model
        parameters $b_n, A_n, k_n$ (bounded as in the previous setting) at
        \emph{each step of the game}. This represents a \emph{fully-adversarial
        and adaptive learning environment}, and unlike the models presented in
        related work~\cite{cartea2017algorithmic}, the source of risk here is
        \emph{exogenous} and \emph{reactive} to the quotes of the MM.
\end{description}

The principle of adversarial learning here --- as with other successful
applications~\cite{goodfellow2014generative} --- is that if the 
MM plays a strategy that is not robust, then this can (and ideally will) be exploited
by the adversary. If a Nash equilibrium (NE) strategy is played by MM then 
their strategy is robust and cannot be exploited. 
While there are no guarantees that a NE will be
reached in our case, we show in Section~\ref{sec:experiments} 
via empirical best response computations that our ARL method does consistently
converge to reasonable approximate NE.
Moreover, we show that we consistently outperform previous approaches
in terms of absolute performance \emph{and} robustness to model ambiguity.

Robustness through ARL was first introduced by~\citet{pinto2017robust} who
demonstrated its effectiveness across a number of standard OpenAI gym domains.
We adapt their RARL algorithm to support incremental actor-critic based methods
and facilitate \emph{asynchronous training};
though many of the features remain the same.  The adversary is trained in
parallel with the market maker, is afforded the same access to state ---
\emph{including MM's inventory} $H_n$ --- and uses the same method for learning,
which is described below.

\subsection{Learning Configuration}
Both agents use the NAC-S($\lambda$) algorithm, a natural actor-critic
method~\cite{thomas2014bias} for stochastic policies (i.e., mixed strategies)
using semi-gradient SARSA($\lambda)$~\cite{rummery1994line} for policy
evaluation. The value functions are represented by
\emph{compatible}~\cite{peters2008natural} radial basis function networks of
100 Gaussian prototypes with \emph{accumulating} eligibility
traces~\cite{sutton2018reinforcement}.

\begin{description}[wide,itemindent=\labelsep]
    \item[States.] The state of the environment $s_n = (t_n, H_n)$ contains only
        the current time $t_n = \frac{nT}{N} = n\Delta t$ and the agent's
        inventory $H_n$, where the transition dynamics are governed by the
        definitions introduced in Section~\ref{sec:model}.
    \item[Policies.] The market maker learns a bivariate Normal policy for
        $\widetilde{p}_n \equiv p_n - S_n$ and $\psi_n$ with diagonal
        covariance matrix. The mean and variance vectors are modelled by linear
        function approximators using 3\textsuperscript{rd}-order polynomial
        bases~\cite{lagoudakis2003least}. 
		Both variances and the mean of the spread term, $\psi_n$, are kept positive via a
		softplus transformation.
        The adversary learns a Beta policy~\cite{chou2017improving},
        shifted and scaled to cover the market parameter intervals. The two
        shape parameters are learnt the same as for the variance of
        the Normal distribution above, with a translation of $+1$ to ensure
        unimodality and concavity.
    \item[Rewards.] The reward function is an adaptation of the optimisation
        objective of \citet{cartea2017algorithmic} and others, for the RL setting of
        incremental, discrete-time feedback:
        \begin{equation}
            R_n = \Delta\Pi_n - \zeta H_n^2 - \begin{cases}
                0 &\enskip \text{for } t < T, \\
                \eta H_n^2 &\enskip \text{otherwise},
            \end{cases}\label{eq:reward}
        \end{equation}
        where $\Pi_n$ refers to the MtM value of the agent's holdings
        (Eq.~\ref{eq:mtm}). This formulation can be either risk-neutral (RN) or
        risk-adverse (RA). For example, if $\eta > 0$ and $\zeta = 0$, then the
        agent is punished if the terminal inventory $H_N$ is non-zero, and the
        solution becomes time-dependent. We refer to this case as RA and the case of
        $\eta = \zeta = 0$ as RN.
\end{description}

\begin{figure*}
    \centering
    \begin{subfigure}[t]{0.48\linewidth}
        \includegraphics[width=1.0\linewidth]{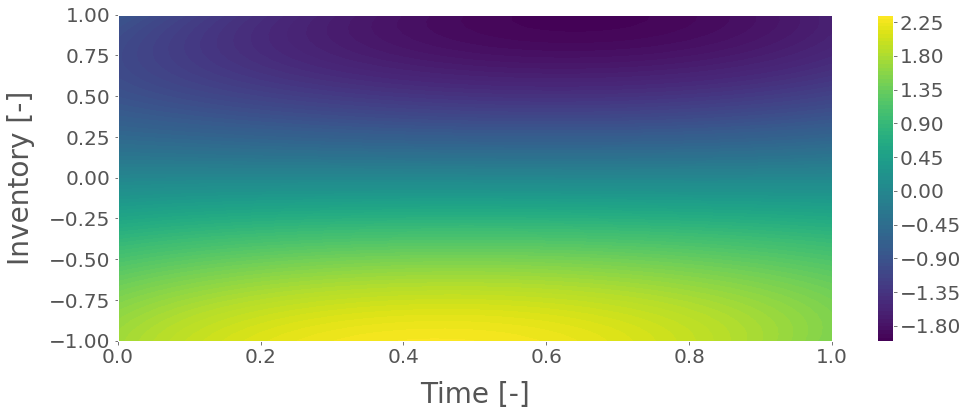}
        \caption{Reservation price offset, $\widetilde{p}(t, H)$.}\label{}
    \end{subfigure}
    \hfill
    \begin{subfigure}[t]{0.48\linewidth}
        \includegraphics[width=1.0\linewidth]{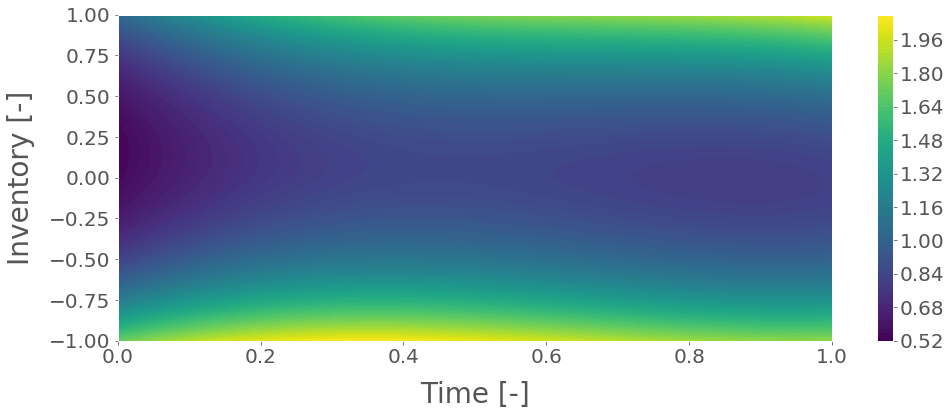}
        \caption{Quoted spread, $\psi(t, H)$.}\label{}
    \end{subfigure}

    \caption{Most probable action for the risk-averse Gaussian policy learnt
    using NAC-S($\lambda$) with $\eta = 0$ and $\zeta =
    0.01$.}\label{fig:ra_policy}
\end{figure*}

\begin{table}
    \centering
    \begin{subtable}[t]{1.0\linewidth}
        \centering
		\scalebox{0.78}{ 
        \begin{tabular}{@{}cc|llll@{}}
            \toprule
            \toprule
            $\eta$ & $\zeta$ & Term.\ wealth & Sharpe & Term.\ inventory & Avg.\ spread \\
            \midrule
            0.0 & 0.0 & $67.0 \pm 12.0$ & $5.57$ & $0.56 \pm 7.55$ & $1.42 \pm 0.02$ \\
            \midrule
            1.0 & 0.0 & $61.3 \pm 8.5$ & $7.25$ & $-0.04 \pm 1.19$ & $1.76 \pm 0.02$ \\
            0.5 & 0.0 & $63.0 \pm 9.9$ & $6.34$ & $0.03 \pm 1.24$ & $1.67 \pm 0.03$ \\
            0.1 & 0.0 & $66.4 \pm 9.4$ & $7.06$ & $-0.16 \pm 1.86$ & $1.42 \pm 0.02$ \\
            0.01& 0.0 & $66.8 \pm 11.1$ & $6.00$ & $0.87 \pm 4.38$ & $1.42 \pm 0.02$ \\
            \midrule
            0.0 & 0.01 & $63.4 \pm 6.8$ & $9.36$ & $0.01 \pm 1.45$ & $1.60 \pm 0.02$ \\
            0.0 & 0.001 & $66.1 \pm 7.4$ & $8.94$ & $0.15 \pm 2.97$ & $1.44 \pm 0.02$ \\
            \bottomrule
            \bottomrule
        \end{tabular}
		}

        \caption{Market makers trained against the \FIXED adversary.}\label{tab:fixed}
    \end{subtable}

    \vspace{1em}
    \begin{subtable}[t]{1.0\linewidth}
        \centering
		\scalebox{0.78}{ 
        \begin{tabular}{@{}cc|llll@{}}
            \toprule
            \toprule
            $\eta$ & $\zeta$ & Term.\ wealth & Sharpe & Term.\ inventory & Avg.\ spread \\
            \midrule
            0.0 & 0.0 & $66.7 \pm 11.8$ & $5.65$ & $0.38 \pm 7.14$ & $1.36 \pm 0.05$ \\
            \midrule
            1.0 & 0.0 & $59.4 \pm 7.6$ & $7.79$ & $0.02 \pm 1.09$ & $1.87 \pm 0.02$ \\
            0.5 & 0.0 & $62.9 \pm 8.4$ & $7.51$ & $-0.04 \pm 1.21$ & $1.68 \pm 0.02$ \\
            0.1 & 0.0 & $65.8 \pm 9.3$ & $7.07$ & $-0.18 \pm 1.57$ & $1.46 \pm 0.03$ \\
            0.01 & 0.0 & $66.6 \pm 9.9$ & $6.69$ & $-0.65 \pm 3.92$ & $1.47 \pm 0.02$ \\
            \midrule
            0.0 & 0.01 & $64.0 \pm 6.7$ & $9.54$ & $-0.06 \pm 1.31$ & $1.60 \pm 0.02$ \\
            0.0 & 0.001 & $65.9 \pm 7.2$ & $9.11$ & $-0.07 \pm 2.62$ & $1.44 \pm 0.02$ \\
            \bottomrule
            \bottomrule
        \end{tabular}
		}

        \caption{Market makers trained against the \RANDOM adversary.}\label{tab:random}
    \end{subtable}

    \vspace{1em}
    \begin{subtable}[h]{1.0\linewidth}
        \centering
		\scalebox{0.78}{ 
        \begin{tabular}{@{}rl|llll@{}}
            \toprule
            \toprule
            & & Term.\ wealth & Sharpe & Term.\ inventory & Avg.\ spread \\
            \midrule
            \multirow{4}{*}{RN} & $b$ & $65.5 \pm 6.9$ & $9.51$ & $0.04 \pm 2.14$ & $1.43 \pm 0.01$ \\
            & $A^\pm$ & $66.9 \pm 13.1$ & $5.11$ & $2.45 \pm 7.76$ & $1.46 \pm 0.02$ \\
            & $k^\pm$ & $66.9 \pm 12.7$ & $5.27$ & $0.52 \pm 7.92$ & $1.45 \pm 0.02$ \\
            & $b, A^\pm, k^\pm$ & $65.1 \pm 6.7$ & $9.78$ & $-0.05 \pm 1.94$ & $1.44 \pm 0.02$ \\
            \midrule
            \multirow{6}{*}{Full}
                & $\eta = 1.0$ & $60.5 \pm 6.8$ & $8.88$ & $-0.02 \pm 0.97$ & $1.75 \pm 0.02$ \\
                & $\eta = 0.5$ & $63.3 \pm 6.8$ & $9.32$ & $-0.07 \pm 1.05$ & $1.60 \pm 0.01$ \\
                & $\eta = 0.1$ & $64.8 \pm 6.7$ & $9.72$ & $-0.06 \pm 1.37$ & $1.49 \pm 0.02$ \\
                & $\eta = 0.01$ & $65.0 \pm 6.7$ & $9.69$ & $-0.08 \pm 1.89$ & $1.49 \pm 0.01$ \\
                & $\zeta = 0.01$ & $62.9 \pm 6.7$ & $9.43$ & $-0.03 \pm 1.19$ & $1.65 \pm 0.02$ \\
                & $\zeta = 0.001$ & $64.3 \pm 6.5$ & $9.85$ & $0.0 \pm 1.71$ & $1.44 \pm 0.01$ \\
            \bottomrule
            \bottomrule
        \end{tabular}
		}

        \caption{Market makers trained against a \STRATEGIC adversary.}\label{tab:strategic}
    \end{subtable}

    \caption{MM performance in terms of the standard desiderata. Each value
        was computed by evaluating the policy over $10^5$ episodes and
        computing the mean and sample standard deviation. Each test episode
        was generated with a \FIXED adversary. Wealth and average spread are
    measured in units of ``currency''; inventory in terms of units of the
``asset''; and the Sharpe ratio is unitless.}\label{tab:performance}
\end{table}

\section{Experiments}\label{sec:experiments}
In each of the experiments to follow, the value function was pre-trained for
1000 episodes (with a learning rate of $10^{-3}$) to reduce variance in early
policy updates. Both the value function and policy were then trained for $10^6$
episodes, with policy updates every 100 time steps, and a learning rate of
$10^{-4}$ for both the critic and policy. The value function was configured to
learn $\lambda = 0.97$ returns. The starting time was chosen uniformly at
random from the interval $t_0 \in [0.0, 0.95]$, with starting price $Z_0 = 100$
and inventory $H_0 \in [\underline{H} = -50, \overline{H} = 50]$. Innovations
in $Z_n$ occurred with fixed volatility $\sigma = 2$ between $[t_0, 1]$ with
increment $\Delta t = 0.005$.

\subsection{Market Maker Desiderata}\label{sec:performance_objectives}
Market makers typically aim to maximise profits while controlling inventory and
quoted spread. Higher terminal wealth with lower variance, less exposure to
inventory risk and smaller spreads are all indicators of the risk aversion (or
lack thereof) of a strategy. This is reflected in the majority of objective
functions for market markers studied in the literature. We consider the
following quantitative desiderata:

\begin{description}[wide,itemindent=\labelsep]
    \item [Profit and loss / Sharpe Ratio.] The distribution of a trading
        strategy's profit and loss is the most fundamental object used to
        measure its performance.  In our setup, we look at an agent's
        distribution of episodic terminal wealth.  In particular, we look at
        the moments $\mathbb{E}[\Pi_N]$ and $\mathbb{V}[\Pi_N]$ of this
        distribution, where the agent would like to maximise the former while
        minimising the latter.  To this end, a common metric used in financial
        literature and in industry is the Shape ratio:
        $\frac{\mathbb{E}[\Pi_N]}{\mathbb{V}[\Pi_N]}$, in essence, the reward
        per unit of risk.\footnote{While the Sharpe ratio would normally use
        returns, for our relative comparison of agents we use terminal
    wealth.} While larger values are better, it is important to note that the
    Sharpe ratio is not a sufficient statistic.

    \item [Terminal inventory.] The distribution of terminal inventory, $H_N$,
        tells us about the robustness of the strategy to adverse price
        movements. In practice, it is desirable for a market maker to
        finish the trading day with small absolute values for $H_N$; though
        this is not always a strict requirement.

    \item [Quoted spread.] The ``competitiveness'' of a market maker 
        is often discussed in terms of the average quoted spread:
        $\mathbb{E}_n[\psi_n]$. Tighter spreads imply more efficient markets and
        exchanges often compensate (with rebates) for smaller spreads. 
\end{description}

\subsection{Results}
\paragraph{\FIXED setting.}
We first trained MMs against the \FIXED adversary; i.e., a standard
single-agent learning environment. Both RN and RA formulations of
Eq.~\ref{eq:reward} were used with risk parameters $\eta \in \{1, 0.5, 0.1,
0.01\}$ and $\zeta \in \{0.01, 0.001\}$. Table~\ref{tab:fixed} summarises the
performance for the resulting agents. 
To provide some further intuition, we illustrate one of the learnt
policies in Figure~\ref{fig:ra_policy}. In this case, the agent learnt to
offset its price asymmetrically as a function of inventory, with increasing
intensity as we approach the terminal time.


\paragraph{\RANDOM setting.}
Next MMs were trained in an environment with a \RANDOM adversary; a simple
extension to the training procedure that aims to develop robustness to
epistemic risk. To compare with earlier results, the strategies were also
tested against the \FIXED adversary --- a summary of which, for the same set
of risk parameters, is given in Table~\ref{tab:random}. The impact on test
performance in the face of model ambiguity was then evaluated by comparing
market makers trained on the \FIXED adversary with those trained against the
\RANDOM adversary. Specifically, out-of-sample tests were carried out in an
environment with a \RANDOM adversary. This means that the model dynamics at
\emph{test-time} were different from those at training time. While not
explicitly adversarial, this misspecification of the model represents a
non-trivial challenge for robustness. Overall, we found that market makers
trained against the \FIXED adversary exhibited no change in average wealth
creation, but an increase of 98.1\% in the variance across all risk
parametrisations. On the other hand, market makers originally trained against
the \RANDOM adversary yielded a lower average 86.0\% increase in the
variance. The \RANDOM adversary helps, but not by much compared with the
\STRATEGIC adversary, as we will see next.


\paragraph{\STRATEGIC setting.}
In this setting, we first consider an adversary that controls the drift $b$
only.  With RN rewards, we found that the adversary learns a (time-independent)
binary policy (Figure~\ref{fig:adversary_policy}) that is identical to the
strategy in the corresponding single-stage game; see
Section~\ref{sec:one_shot}. We also found that the strategy learnt by the MM in
this setting generates profits and associated Sharpe ratios in excess of all
other strategies seen thus far when tested against the \FIXED adversary (see
Table~\ref{tab:strategic}). This is true also when comparing with tests run
against the \RANDOM or \STRATEGIC adversaries, suggesting that the
adversarially trained MM is indeed more robust to test-time model
discrepancies.

\begin{figure}[tb]
    \centering
    \includegraphics[width=1.0\linewidth]{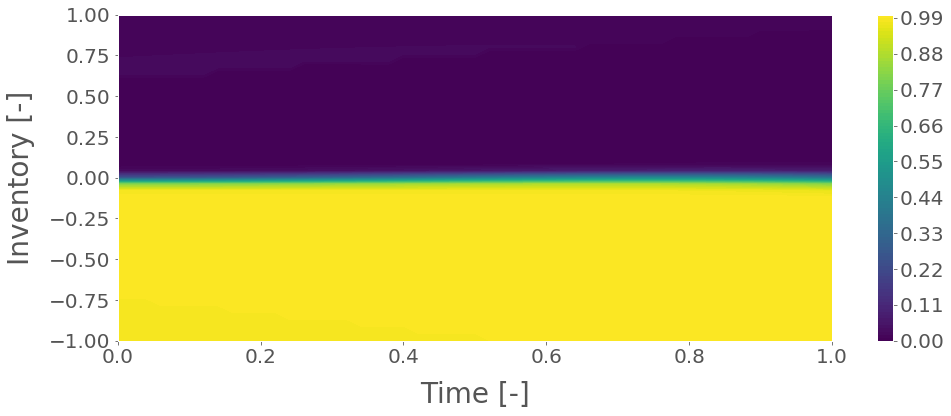}

    \caption{Policy learnt by the adversary for manipulating price against the
    market maker with $\eta = \zeta = 0$.}\label{fig:adversary_policy}
\end{figure}

This, however, does not extend to \STRATEGIC
adversaries with control over either \emph{only} $A^\pm$ or \emph{only}
$k^\pm$. In these cases, performance was found to be \emph{no better} than
the corresponding MMs trained in the \FIXED setting with a
conventional learning setup. The intuition for this again derives from the
single-stage analysis. That is, the adversary almost surely chooses a strategy
that minimises $A^\pm$ (equiv.\ maximises $k^\pm$) in order to decrease the
probability of execution, thus decreasing MM's profits derived from execution and
MM's ability to manage inventory effectively. The
control afforded to the adversary must correlate in some way with
sources of variance --- such as inventory speculation --- in order for
robustness to correspond to practicable forms of risk-aversion.

Then the  natural question is then whether an adversary with simultaneous
control over $b, A^\pm$, and $k^\pm$ produces strategies outperforming those
where the adversary controls $b$ alone. This is plausible since combining
the three model parameters can lead to more interesting strategies, e.g.,
driving inventory up/down only to increase drift at the peak (i.e., pump and
dump).
We investigated by training an adversary with control over all
three parameters. The resulting performance can be found in
Table~\ref{tab:strategic}, which shows an improvement in the Sharpe ratio of
0.27 and lower variance on terminal wealth.  Interestingly, these MMs also
quote tighter spreads on average --- the values even approaching that of the
risk-neutral MM trained against a \FIXED adversary. This indicates that the
strategies are able to achieve epistemic risk aversion \emph{without} charging
more to counterparties.




Exploring the impact of varying risk parameters 
of the reward function, $\eta$ and
$\zeta$ (see~\eqref{eq:reward}), we found that in all cases MM strategies
trained against a \STRATEGIC adversary with an RA reward outperformed their
counterparts in Tables~\ref{tab:fixed}~and~\ref{tab:random}. It is unclear,
however, if changes to the reward function away from the RN variant actually
improved the strategy in general. Excluding when $\zeta = 0.001$, all values
appear to do worse than for an adversarially trained MM with RN reward. It may
well be that the addition of inventory penalty terms actually boosts the power
of the adversary and results in strategies that try to avoid trading at all, a
frequent problem in this domain.



\paragraph{Verification of approximate equilibria.}
Holding the strategy of one player fixed, we empirically computed the best
response against it by training. We found consistently that neither the trader
nor adversary deviated from their policy. This suggests that our ARL method were
finding reasonable approximate Nash equilibria. While we do not provide a full
theoretical analysis of the stochastic game, these findings are corroborated by
those in Section~\ref{sec:one_shot}, since, as seen in
Figure~\ref{fig:adversary_policy}, the learned policy in the multi-stage setting
corresponds closely to the equilibrium strategy from the single-stage case that
was presented in Section~\ref{sec:one_shot}.

\section{Conclusions}
We introduced a new approach for learning trading strategies with ARL.
The learned strategies are robust to the discrepancies between the market model
in training and testing. 
We show that our approach leads to
strategies that outperform previous methods in terms of PnL and Sharpe ratio,
and have comparable spread efficiency. This is shown to be the case for
out-of-sample tests in all three of the considered settings, \FIXED, \RANDOM,
and \STRATEGIC. 
In other words, our learned strategies are not only more robust to
misspecification, but also dominate in overall performance. 

In some special cases we show that the learned strategies correspond to Nash
equilibria of the corresponding single-stage game. More widely, we empirically
show that the learned strategies correspond to approximate equilibria in the
multi-stage stochastic game. 

In future we plan to:
\begin{enumerate}[leftmargin=0.4cm]
    \item Extend to oligopolies of market makers.
    \item Apply to data-driven and multi-asset models.
    \item Further explore existence of equilibria and design provably-convergent algorithms.
\end{enumerate}

Finally, we remark that, while our paper focuses on market making, the approach
can be applied to other trading scenarios such as optimal execution and
statistical arbitrage, where we believe it is likely to offer similar benefits.
Further, it is important to acknowledge that this methodology has significant
implications for safety of RL in finance. Training strategies that are
explicitly robust to model misspecification makes deployment in the real-world
considerably more practicable.

\paragraph{Software.}
All our code is freely accessible on GitHub:
\url{https://github.com/tspooner/rmm.arl}.

\small
\bibliographystyle{named}
\bibliography{ml,maths,finance}

\begin{thebibliography}{}

\bibitem[\protect\citeauthoryear{Abergel \bgroup \em et al.\egroup
  }{2016}]{abergel2016limit}
F. Abergel, M. Anane, A. Chakraborti, A. Jedidi, and I.~M. Toke.
\newblock {\em {L}imit {O}rder {B}ooks}.
\newblock {C}ambridge {U}niversity {P}ress, 2016.

\bibitem[\protect\citeauthoryear{Abernethy and Kale}{2013}]{AbernethyK13}
J.~D. Abernethy and S. Kale.
\newblock {A}daptive {M}arket {M}aking via {O}nline {L}earning.
\newblock In {\em Proc.\ of {NIPS}}, pages 2058--2066, 2013.

\bibitem[\protect\citeauthoryear{Avellaneda and
  Stoikov}{2008}]{avellaneda2008high}
M. Avellaneda and S. Stoikov.
\newblock {H}igh-frequency trading in a limit order book.
\newblock {\em {Q}uantitative {F}inance}, 8(3):217--224, 2008.

\bibitem[\protect\citeauthoryear{Bellemare \bgroup \em et al.\egroup
  }{2017}]{bellemare2017distributional}
M.~G. Bellemare, W. Dabney, and R. Munos.
\newblock {A} {D}istributional {P}erspective on {R}einforcement {L}earning.
\newblock In {\em Proc.\ of ICML}, pages 449--458, 2017.

\bibitem[\protect\citeauthoryear{Cartea and Jaimungal}{2015}]{cartea2015order}
{\'A}. Cartea and S. Jaimungal.
\newblock {O}rder-{F}low and {L}iquidity {P}rovision.
\newblock {\em {A}vailable at SSRN 2553154}, 2015.

\bibitem[\protect\citeauthoryear{Cartea \bgroup \em et al.\egroup
  }{2015}]{cartea2015algorithmic}
{\'A}. Cartea, S. Jaimungal, and J. Penalva.
\newblock {\em {A}lgorithmic and {H}igh-{F}requency {T}rading}.
\newblock {C}ambridge {U}niversity {P}ress, 2015.

\bibitem[\protect\citeauthoryear{Cartea \bgroup \em et al.\egroup
  }{2017}]{cartea2017algorithmic}
{\'A}. Cartea, R. Donnelly, and S. Jaimungal.
\newblock {A}lgorithmic {T}rading with {M}odel {U}ncertainty.
\newblock {\em SIAM {J}ournal on {F}inancial {M}athematics}, 8(1):635--671,
  2017.

\bibitem[\protect\citeauthoryear{Chan and Shelton}{2001}]{chan2001electronic}
N.~T. Chan and C. Shelton.
\newblock {A}n {E}lectronic {M}arket-{M}aker.
\newblock 2001.

\bibitem[\protect\citeauthoryear{Chou \bgroup \em et al.\egroup
  }{2017}]{chou2017improving}
P.-W. Chou, D. Maturana, and S. Scherer.
\newblock {I}mproving {S}tochastic {P}olicy {G}radients in {C}ontinuous
  {C}ontrol with {D}eep {R}einforcement {L}earning {U}sing the {B}eta
  {D}istribution.
\newblock In {\em Proc.\ of ICML}, pages 834--843, 2017.

\bibitem[\protect\citeauthoryear{Fodra and Labadie}{2012}]{fodra2012high}
P. Fodra and M. Labadie.
\newblock {H}igh-frequency market-making with inventory constraints and
  directional bets.
\newblock {\em arXiv preprint arXiv:1206.4810}, 2012.

\bibitem[\protect\citeauthoryear{Garc{\i}a and
  Fern{\'a}ndez}{2015}]{garcia2015comprehensive}
J. Garc{\i}a and F. Fern{\'a}ndez.
\newblock {A} comprehensive survey on safe reinforcement learning.
\newblock {\em JMLR}, 16(1):1437--1480, 2015.

\bibitem[\protect\citeauthoryear{Glosten and Milgrom}{1985}]{glosten1985bid}
L.~R. Glosten and P.~R. Milgrom.
\newblock {B}id, {A}sk and {T}ransaction {P}rices in a {S}pecialist {M}arket
  with {H}eterogeneously {I}nformed {T}raders.
\newblock {\em {J}ournal of {F}inancial {E}conomics}, 14(1):71--100, 1985.

\bibitem[\protect\citeauthoryear{Goodfellow \bgroup \em et al.\egroup
  }{2014}]{goodfellow2014generative}
I. Goodfellow, J. Pouget-Abadie, M. Mirza, B. Xu, D. Warde-Farley, S. Ozair, A.
  Courville, and Y. Bengio.
\newblock {G}enerative {A}dversarial {N}ets.
\newblock In {\em Proc.\ of NIPS}, pages 2672--2680, 2014.

\bibitem[\protect\citeauthoryear{Gould \bgroup \em et al.\egroup
  }{2013}]{gould2013limit}
M.~D. Gould, M.~A. Porter, S. Williams, M. McDonald, D.~J. Fenn, and S.~D.
  Howison.
\newblock {L}imit {O}rder {B}ooks.
\newblock {\em {Q}uantitative {F}inance}, 13(11):1709--1742, 2013.

\bibitem[\protect\citeauthoryear{Grossman and
  Miller}{1988}]{grossman1988liquidity}
S.~J. Grossman and M.~H. Miller.
\newblock {L}iquidity and {M}arket {S}tructure.
\newblock {\em {T}he {J}ournal of {F}inance}, 43(3):617--633, 1988.

\bibitem[\protect\citeauthoryear{Gu{\'e}ant and Manziuk}{2019}]{gueant2019deep}
O. Gu{\'e}ant and I. Manziuk.
\newblock {D}eep reinforcement learning for market making in corporate bonds:
  beating the curse of dimensionality.
\newblock {\em arXiv preprint arXiv:1910.13205}, 2019.

\bibitem[\protect\citeauthoryear{Gu{\'e}ant \bgroup \em et al.\egroup
  }{2011}]{gueant2013dealing}
O. Gu{\'e}ant, C.-A. Lehalle, and J. Fernandez-Tapia.
\newblock {D}ealing with the {I}nventory {R}isk: {A} solution to the market
  making problem.
\newblock {\em {M}athematics and {F}inancial {E}conomics}, 7(4):477--507, 2011.

\bibitem[\protect\citeauthoryear{Gu{\'e}ant}{2017}]{gueant2017optimal}
O. Gu{\'e}ant.
\newblock {O}ptimal market making.
\newblock {\em {A}pplied {M}athematical {F}inance}, 24(2):112--154, 2017.

\bibitem[\protect\citeauthoryear{Ho and Stoll}{1981}]{ho1981optimal}
T. Ho and H.~R. Stoll.
\newblock {O}ptimal {D}ealer {P}ricing {U}nder {T}ransactions and {R}eturn
  {U}ncertainty.
\newblock {\em {J}ournal of {F}inancial {E}conomics}, 9(1):47--73, 1981.

\bibitem[\protect\citeauthoryear{Lagoudakis and
  Parr}{2003}]{lagoudakis2003least}
M.~G. Lagoudakis and R. Parr.
\newblock {L}east-{S}quares {P}olicy {I}teration.
\newblock {\em JMLR}, 4:1107--1149, 2003.

\bibitem[\protect\citeauthoryear{Littman}{1994}]{littman1994markov}
M.~L. Littman.
\newblock {M}arkov games as a framework for multi-agent reinforcement learning.
\newblock In {\em Proc.\ of {ICML}}, pages 157--163. 1994.

\bibitem[\protect\citeauthoryear{Maddison \bgroup \em et al.\egroup
  }{2017}]{madisson2017particle}
C.~J. Maddison, D. Lawson, G. Tucker, N. Heess, M. Norouzi, A. Mnih, A. Doucet,
  and Y.~W. Teh.
\newblock {P}article {V}alue {F}unctions.
\newblock In {\em ICLR 2017 {W}orkshop {P}roceedings}, 2017.

\bibitem[\protect\citeauthoryear{Nevmyvaka \bgroup \em et al.\egroup
  }{2006}]{nevmyvaka2006reinforcement}
Y. Nevmyvaka, Y. Feng, and M. Kearns.
\newblock {R}einforcement {L}earning for {O}ptimized {T}rade {E}xecution.
\newblock In {\em Proc. of ICML}, pages 673--680. ACM, 2006.

\bibitem[\protect\citeauthoryear{Othman}{2012}]{othman}
A. Othman.
\newblock {\em {A}utomated {M}arket {M}aking: {T}heory and {P}ractice}.
\newblock PhD thesis, {CMU}, 2012.

\bibitem[\protect\citeauthoryear{P{\'e}rolat \bgroup \em et al.\egroup
  }{2018}]{perolat2018actor}
J. P{\'e}rolat, B. Piot, and O. Pietquin.
\newblock {A}ctor-{C}ritic {F}ictitious {P}lay in {S}imultaneous {M}ove
  {M}ultistage {G}ames.
\newblock In {\em Proc.\ of AISTATS}, 2018.

\bibitem[\protect\citeauthoryear{Peters and Schaal}{2008}]{peters2008natural}
J. Peters and S. Schaal.
\newblock {N}atural {A}ctor-{C}ritic.
\newblock {\em {N}eurocomputing}, 71(7-9):1180--1190, 2008.

\bibitem[\protect\citeauthoryear{Pinto \bgroup \em et al.\egroup
  }{2017}]{pinto2017robust}
L. Pinto, J. Davidson, R. Sukthankar, and A. Gupta.
\newblock {R}obust {A}dversarial {R}einforcement {L}earning.
\newblock In {\em Proc.\ of ICML}, volume~70, pages 2817--2826, 2017.

\bibitem[\protect\citeauthoryear{Rajeswaran \bgroup \em et al.\egroup
  }{2017}]{rajeswaran2016epopt}
A. Rajeswaran, S. Ghotra, B. Ravindran, and S. Levine.
\newblock {EPO}pt: {L}earning {R}obust {N}eural {N}etwork {P}olicies using
  {M}odel {E}nsembles.
\newblock In {\em Proc. of ICLR}, 2017.

\bibitem[\protect\citeauthoryear{Rummery and Niranjan}{1994}]{rummery1994line}
G.~A. Rummery and M. Niranjan.
\newblock {O}n-line {Q}-learning {U}sing {C}onnectionist {S}ystems.
\newblock Technical report, {D}epartment of {E}ngineering, {U}niversity of
  {C}ambridge, 1994.

\bibitem[\protect\citeauthoryear{Sion}{1958}]{sion1958general}
M. Sion.
\newblock {O}n {G}eneral {M}inimax {T}heorems.
\newblock {\em {P}acific {J}ournal of {M}athematics}, 8(1):171--176, 1958.

\bibitem[\protect\citeauthoryear{Spooner \bgroup \em et al.\egroup
  }{2018}]{spooner2018market}
T. Spooner, J. Fearnley, R. Savani, and A. Koukorinis.
\newblock {M}arket {M}aking via {R}einforcement {L}earning.
\newblock In {\em Proc.\ of AAMAS}, pages 434--442, 2018.

\bibitem[\protect\citeauthoryear{Sutton and
  Barto}{2018}]{sutton2018reinforcement}
R.~S. Sutton and A.~G. Barto.
\newblock {\em {R}einforcement {L}earning: {A}n {I}ntroduction}.
\newblock MIT {P}ress, 2018.

\bibitem[\protect\citeauthoryear{Tamar \bgroup \em et al.\egroup
  }{2012}]{tamar2012policy}
A. Tamar, D. Di~Castro, and S. Mannor.
\newblock {P}olicy {G}radients with {V}ariance {R}elated {R}isk {C}riteria.
\newblock In {\em Proc.\ of ICML}, pages 1651--1658, 2012.

\bibitem[\protect\citeauthoryear{Thomas}{2014}]{thomas2014bias}
P. Thomas.
\newblock {B}ias in {N}atural {A}ctor-{C}ritic {A}lgorithms.
\newblock In {\em Proc.\ of ICML}, volume~32, pages 441--448, 2014.

\bibitem[\protect\citeauthoryear{Vyetrenko and Xu}{2019}]{vyetrenko2019risk}
S. Vyetrenko and S. Xu.
\newblock {R}isk-{S}ensitive {C}ompact {D}ecision {T}rees for {A}utonomous
  {E}xecution in {P}resence of {S}imulated {M}arket {R}esponse.
\newblock {\em arXiv preprint arXiv:1906.02312}, 2019.

\end{thebibliography}

\end{document}